\documentclass[final,1p,times]{elsarticle}

\usepackage{amsthm}
\usepackage{amssymb}

\journal{arxiv.org}

\newtheorem{theorem}{Theorem}

\newtheorem{lemma}{Lemma}
\newtheorem{corollary}[theorem]{Corollary}
\newdefinition{remark}{Remark}
\newdefinition{example}{Example}
\newdefinition{definition}{Definition}

\begin{document}

\begin{frontmatter}



\title{A collective of stateless automata in a $n$-dimensional environment as a distributed dynamic automaton-like object: a model and its corollaries}


\author{Oleksiy Kurganskyy}
\ead{kurgansk@gmx.de}
\address{Institute of Applied Mathematics and Mechanics,\\
National Academy of Science of Ukraine,\\
Roza Luksemburg st., 74, 83114 Donetsk, Ukraine}

\begin{abstract}
In this work a collective of interacting stateless automata in a discrete geometric $n$-dimenstional environment is considered as an integral automaton-like computational dynamic object. For such distributed on the environment object different approaches to definition of the measure of state transition are possible. We propose an approach for defining what a state is. The approach is based on the concept of relativity in Poincar\'e's interpretation.
\end{abstract}


\end{frontmatter}

\section{Introduction}

Currently there is a great interest in computational models consisting of underlying regular computational environments, and distributed computational structures built on them. Examples of such models are cellular automata, spacial computation and space-time crystallography~\cite{pedestrian}. For any computational model it is natural to compare functional, algorithmic and structural properties of different but related computational structures. In the finite automata theory an example of such comparison is automata homomorphism and, in particular, automata isomorphism. If we keep to the finite automata theory, a fundamental question what a state of a distributed computational structure is arises. This work is devoted to a particular solution of the issue.

The work consists of an informal presentation of the background idea of what we mean by computation of a distributed in an environment algorithmic structure, and an illustration of this idea by a simple computational model with a regular and discrete dynamics that is designed specifically for illustration purposes. The model and the problem statement are similar to the model of~\cite{pedestrian}, but differ from it in essence. One of the distinguishing features of the model is dynamics of computational structure in the environment. This produces a number of results concerning the relationship between computational and dynamic properties of these structures.

\section{Background idea}

In this work the collectives of stateless (i.e. with one state) automata interacting with an environment defined as a graph are considered. We study a collective of automata as an integral automaton-like dynamic computational object distributed in an environment. The fundamental question what is the state of such dispersed and moving on the environment object and how to measure the amount of state transitions is quite non-trivial. As opposed to the finite state automata where the measure of state transition is one state per unit of time, for a computational dynamic object distributed on the environment different approaches to definition of the measure of state transition are certainly possible. 


The idea of our approach came from the special relativity theory. It is based on the concept of relativity in Poincar\'e's interpretation~\cite{poincare}. In explanation of how to generally understand the relativity Poincar\'e begins with an example of resizing of dimensions in the Universe by the same number of times and proceeds with considering arbitrary deformations concluding that they should be unnoticed by any observer because his standards are subject to the same deformations. This reasoning coupled with the principle that the process of computation in the object is not possible without any changes in it is used in this study to define a state of collective automata. A ``change'' in an object is a change in the relative position of its ``elementary'' parts. Thus, the movement in the environment underlies the process of computation in our model.

Let us explain it by an example using pawns on a chessboard, see Fig.~\ref{shess}. The chessboard is provided with a natural reference frame. Suppose that we can move any pawn one chess square per unit of time in one of four directions: $\leftarrow$, $\uparrow$, $\rightarrow$, $\downarrow$, i.e. pawn's velocity is one chess square in a certain direction per unit of time. Let us compose a figure from the pawns, for example, an ``O''-like figure, and look at them as an integral object. Let us define the velocity of the object on the chessboard as the average velocity of its pawns. Suppose that the object is moving at maximal velocity ``one chess square per unit of time'' in a constant direction. Can the object be transformed simultaneously with the motion from ``O'' to, for example, ``T''? It is obvious that it can not.  That is, at maximum constant velocity in the example the object cannot be changed and, from our point of view, its state is invariable and it performs no computation. This point of view is formally illustrated in this work by the simplest example model of stateless automata interacting with $n$-dimensional environment. Note that the figure~\ref{shess} is not quite correct from our point of view as pawns have two states: good and bad moods. Easy to see that on the third chessboard letter ``O'' differs from the first two. We consider only the case of stateless pawns.

The introduced illustrative model is computationally universal, and collectives of automata in the environment can be seen as automaton-like computational objects. By analogy with Turing machines, which can answer certain questions about properties of words on the tapes, natural questions arise for these objects, such as what properties of the environment and other objects in it they can identify. One of the interesting questions is what can an object say about the velocity of its elementary parts (i.e. stateless automata). Can it ``perceive'' any changes in velocity of elementary parts it consists of? This question is similar to the issue in the Poincar\'e's story about relativity: can the observer see the deformation of the space, which includes the deformation of measurement standards? Having the answer ``no'' as a goal, we define our computational model. This goal determines the language (motion velocity, proper time velocity as a measure of state transition, reference frame) of interaction between collectives of automata. 

To emphasize a physical analogy in the proposed model and the problem statement we use the short word ``body'' as alias for ``collective of automata''.

This work develops~\cite{kurg2010}.

\begin{figure}
\centering
\includegraphics[scale=0.4]{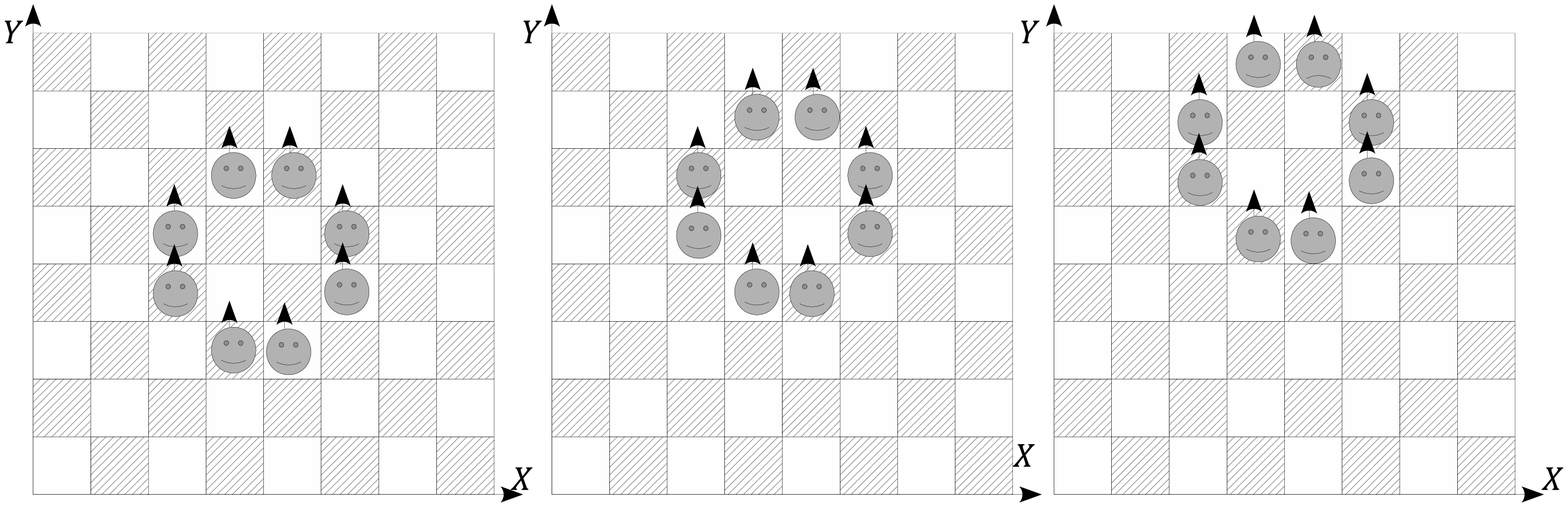}
\caption{A chessboard with pawns}
\label{shess}
\end{figure}

\section{Definitions}

The computational model used in this study consists of two main components: an underlying environment $G$ that is represented by a graph and a set of stateless automata interacting with the environment.

The environment $G$ is an infinite or finite directed graph defined as follows. Let $D$ be a finite set that we call the set of \emph{directions}. Without loss of generality, we assume that all directions are numbered by integers from $1$ to $|D|$. We associate with each arc in the graph a direction from $D$. For each arc $e=(x,y)$ there exists an arc $e'=(y,z)$ of the same direction. All arcs ending in the same vertex have different directions. All arcs that start in the same vertex have different directions. If different arcs $e$ and $e'$ end in the same vertex we will say that they are intersecting arcs. The neighborhood of an arc $e$ is understood to be the set of all $e$ intersecting arcs.

Let the graph $G$ be embedded in an $n$-dimensional affine metric space $E$ as follows. Each arc is a line segment of length $1$. Adjacent arcs of the same direction lie on the same line. Thus the set of directions $D$ takes on the meaning of the set of $n$-dimensional vectors in $E$. We call the set of vectors $D$ the set of \emph{actual spacial directions}. Let us fix the origin of the space $E$ so that it coincides with some vertex of the graph. Thus, each vertex and each point of the arcs get an $n$-dimensional coordinate in the space $E$. The space $E$ is said to be absolute, and the coordinates in it absolute spacial coordinates.

As usual, let $A=(S_A,I_A,O_A,{\delta}_A, {\lambda}_A)$ be a Mealy automaton, where $S_A$, $I_A$ and $O_A$ are the sets of states, input symbols, and output symbols, respectively, and ${\delta}_A: S_A \times I_A \rightarrow S_A$ and ${\lambda}_A: S_A \times I_A \rightarrow O_A$  are transition function and function of outputs respectively. We consider only stateless Mealy automata. The set of states of a stateless automaton consists of a single state, so there is no sense in mentioning its transition function and its set of states. Thus we write $A=(I_A,O_A,\lambda_A)$ instead of $A=(S_A,I_A,O_A,{\delta}_A, {\lambda}_A)$.  Within the framework of this article for reasons of consistency of latter definitions we name a stateless Mealy automaton an \emph{elementary body}, and we name its unique state as the \emph{internal state}. The elementary bodies will be denoted by lowercase letters, for example, $b=(I_b,O_b,\lambda_b)$. We assume also that elementary bodies are coloured in a way that isomorphic automata will have the same colour and non-isomorphic automata will have different colours. We assume that $r$ different colours numbered from $1$ to $r$ are used. Every integral moment of time $t$, that we call absolute, any elementary body $b$ is located on an arc $b(t)$ of the environment $G$. The input for an elementary body $b$ is the ordered set $\{p_{ij}\}_{1\leq i\leq |D|,1\leq j\leq r}\in I_b$ called the \emph{neighbourhood state} of the arc $b(t)$, where $p_{ij}$ are the number of all elementary bodies of the colour $j$ located on the arc of direction $i$ that ends in the same vertex as $b(t)$ at the moment of time $t$, $1\leq i\leq |D|$, $1\leq j\leq r$. The definition does not imply that the input set is finite. We have done nothing to circumvent this problem but we can simply assume that elementary bodies interact in such a way that the set of all possible input symbols can only be finite. The output of an elementary body is a direction from $D$. If the output of an elementary body $b$ at a moment of time $t$ is the direction $i$ and the arc $b(t)$ ends in the vertex $x$, then at the next moment of time the arc $b(t+1)$ starts from $x$ and has direction $i$. If the arcs $b(t)$ and $b(t+1)$ are of the same direction we say that $b$ does not change its \emph{external state} at the moment of time $t$. Otherwise we say that $b$ changes its \emph{external state}. If $b$ does not change its external state we also say that $b$ moves rectilinearly. Additionally we assume that each elementary body can not change its external state if all intersecting arcs are empty, that is, the interaction between elementary bodies occurs only by collisions in the vertices of the environment (compare with the notion of vacuum state in~\cite{pedestrian}). The elementary bodies can be seen as analogues of signals propagating in the causal network~\cite{pedestrian}. Propagation of signals in~\cite{pedestrian} depends on the functions in the nodes of a causal network, in our model it depends on the output functions $\lambda$ of elementary bodies, i.e., on the properties of ``signal''.

Let us represent discrete dynamics of elementary bodies in the graph $G$ by the continuous dynamics in $E$ as follows. Let $b$ be at an integral moment of time $t$ on the arc $b(t)=(v_0,v_1)$. Let the $n$-dimensional absolute spacial coordinates of the vertices $v_0$ and $v_1$ be $\vec{x}_0$ and $\vec{x}_1$, respectively. Then the elementary body $b$ at time $t + \lambda$, $0\leq\lambda<1$, has the absolute spacial coordinate $\vec{x}_0 + (\vec{x}_1-\vec{x}_0)\cdot\lambda$. We denote by $x_b(t)$ the absolute spacial coordinate of $b$ at time $t$.

We denote by $\tau_b=\tau_b(t)$ a measure of external state transition of $b$ until the moment of time $t$. By definition, if the elementary body $b$ moved rectilinearly from a moment of time $t_1$ to $t_2$ , then $\tau_b(t_1) = \tau_b(t_2)$. Definition of functional behavior of $\tau_b$ when $b$ moves nonrectilinearly requires additional considerations that we provide below.

We call $\tau = \tau_b(t)$ the \emph{proper time} of $b$ and $w_b(t)=\tau_b(t+1)-\tau_b(t)$ the \emph{proper time velocity} of $b$. We call $w_b(t)$ uniform proper time velocity if $w_b(t)$ is a constant. We denote by $v_b(t)=x_b(t+1)-x_b(t)$ the \emph{absolute spacial velocity} of $b$ at the moment of time $t$. We call it uniform spacial velocity if $v_b(t)$ is a constant.

We call the pair of a space coordinate $x$ and a time coordinate $t$ as (spacial-time) coordinate in the \emph{absolute reference frame} $O$. We shall call $O$ the event space as well.


In addition to absolute reference frame $O$ we introduce the notion of \emph{absolute actual reference frame} $Q$ as follows. Let $X$ be the set of spacial coordinates of all vertices of graph $G$. We construct a graph $GT$ with the set of vertices $X\times Z$ such that there exists an arc from a vertex $(x_1, t_1)$ to a vertex $(x_2, t_2)$ if and only if the arc $(x_1,x_2)$ belongs to the graph $G$ and $t_2 = t_1 +1$. Thus, the dynamics of an elementary body in the $O$ is the dynamics on the graph $GT$. Let $D$ be the set of vectors $\{\vec{1},\vec{2},\ldots,\vec{m}\}$ in the space $E$. Let us denote $e_i=(\vec{i},1)$, $\vec{i}\in D$, $1\le i\le m$. We call the ordered set $\{e_i|1\le i\le m\}$ the set of actual space-time directions, or simply the set of \emph{actual directions}.

\begin{lemma}\label{lemma:1}
Let an elementary body $b$ be in the origin of space $E$ at the time $0$. Then a space-time coordinate $(x_b(t),t)$ of $b$ at a moment of time $t$ can be represented as a linear combination of actual directions.
\end{lemma}
\begin{proof}
The proof is obvious because elementary body moves only in actual directions in the event space $O$.
\end{proof}

Note that the definitions do not imply the linear independence of actual directions. But we say that the coefficients of linear combination of actual directions that forms a vector $(x_b(t),t)$ are coordinates of $(x_b(t),t)$ in the absolute actual reference frame $Q$. By definition, we assume that the dimensions of linear spaces $O$ and $Q$ are equal. 

Let us proceed to considering the collectives of elementary bodies.

\begin{definition}
A body is an arbitrary finite set of elementary bodies.
\end{definition}

According to the defintion different bodies may have common parts and one body can contain another body as a subset.
If an elementary body belongs to a body then we will consider it as an elementary part of this body. An elementary body can be an elementary part of different bodies simultaneously.

Let a body $B$ consist of $n$ elementary bodies enumerated by $\{1,2, \ldots, k\}$. Then the absolute (average) coordinate of the body $B$ at time $t$ is the value $x_B(t)=\frac{x_1(t)+ \ldots + x_k(t)}{k}$ and absolute spacial velocity of $B$ at time $t$ is the value $v_B(t)=x_B(t+1)-x_B(t)$. It follows from the definitions that the maximum modulus of spacial velocity of the bodies is $1$.

\section{External state}

A body interacting with other bodies influences them and is also under their influence. It is quite natural to describe such influences on the basis of the notion of a state of a body. Our definition of a state of a body takes into consideration the relative position of its elementary parts in the environment. The changes of relative position of elementary parts in a body can affect the whole body or its part. This motivates the question how to measure the amount of state transition. Before defining the notion of a state, that we will call the external state and that generalises the notion of the external state of  elementary body, we introduce the denotation for the measure $\tau = \tau_B(t)$ of external state transition of a body $B$ with the flow of absolute time $t$. A casual meaning of $\tau = \tau_B(t)$ is the ``age'' of the body $B$ at the moment $t$. We call $\tau = \tau_B(t)$ the proper time of $B$.

Independently from the definition of $\tau = \tau_B(t)$, we introduce the velocity $w_B(t)$ of the proper time of $B$ as $w_B(t)=\tau_B(t+1)-\tau_B(t)$. We call this value as the proper time velocity of $B$ at the moment of the absolute time $t$.

\begin{definition}
For any body $B$ $w_B(t)=0 \Leftrightarrow \forall_{b \in B}w_b(t)=0$
\end{definition}

\begin{definition}
If $w_B(t)=0$ then we say that the body $B$ does not change its external state at the moment of time $t$.
\end{definition}

It follows from this that a body $B$ does not change its external state if all its elementary bodies do not change their external states. It means that two bodies are at the same external state in the environment if one of them can be transformed into another by 
straight-line shifts on the equal number of steps applied to all its elementary parts in direction corresponding to their external states.

Note that the definition does not forbid a situation when the body has zero absolute spacial velocity and zero absolute proper time velocity simultaneously, see~\cite{kurg2010}.

\begin{theorem}
For any body $B$, if $|v_B(t)|=1$ then $w_B(t)=0$.
\end{theorem}
\begin{proof}
The statement follows from the fact that any change of the external state of a body 
is not possible in case of maximal spacial velocity of all its elementary parts.
\end{proof}

We have defined the meaning of two bodies being in the same external state, rather than what the external state of a body in fact is. If needed the notion of external state can be generally defined as follows: since the relation ``to be in the same external state'' is an equivalence relation, the external states are equivalence classes of this relation. The same holds for the latter definition of internal state.

Definition of functional behavior of $\tau_B$ when $B$ changes its external state will be given below.

\section{Internal state}

The notion of external state of a body allows to start considering the bodies as an automata-like model of algorithms. It is natural to ask a functional equivalence of different bodies for example something like automata isomorphism in the finite automata theory. But because two bodies with different absolute spacial velocities are definitely in different external states we can not compare them functionally. For example there is no sense to ``ask'' a body to determine its absolute spacial velocity. However we would like to identify two bodies as the same algorithm even if they move with different spacial velocities. It will be achieved by introduction of affine isomorphism of bodies through definition of inertial reference frame associated with a body so that the external state of a body will be presented as a pair of components: spacial velocity of the body and its spacial velocity invariant internal state. The point of introducing the notion of inertial  reference frame associated with a body lies in the ability to consider other bodies in relation to the given one. With reference frames we attempt to develop a language of interaction between bodies just as the input and output alphabets of finite Mealy automata are for the interaction between them. The language that we develop is one of the many possible and thus our approach reflects a Poincar\'e's conventional point of view on the physical laws. An example of inertial reference frame is the absolute reference frame $O$ associated with an immovable body $B$ such that for all $t$ $x_B(t)=0$, $v_B(t)=0$, $w_B(t)=1$, $\tau_B(0)=0$, and, hence, $\tau_B(t)=t$. Thus, the introduced notions of absolute time, absolute coordinate and absolute spacial velocity implicitly mean an absolutely motionless body in relation to which objects were considered. The reference frames associated with the bodies allow us to make these notions relative.

Let us denote (for a pair of bodies $A$ and $B$) by $x_{AB}(\tau_B)$, $v_{AB}(\tau_B)$, $w_{AB}(\tau_B)$ and $\tau_{AB}(\tau_B)$ the coordinate, the spacial velocity, the proper time velocity and the proper time of the body $A$ at the moment of time $\tau_{B}$ in the reference frame $O_B$ associated with the body $B$, respectively. By definition we assume that $x_{BB}(\tau_B)\equiv 0$, $v_{BB}(\tau_B)\equiv 0$, $w_{BB}(\tau_B)\equiv 1$ and $\tau_{BB}(\tau_B)=\tau_B$. 

\begin{definition}
A body $B$ is called an inertial body if $v_B(t)$ and $w_B(t)$ are both constants.
\end{definition}

The property to be inertial implies uniform changes of not only spacial coordinates but also time coordinates. For the sake of we shall only simplicity consider the case of inertial bodies.

\begin{definition}
A reference frame associated with an inertial body will be called an inertial reference frame.
\end{definition}

The only restriction imposed on the inertial reference frames is the property that space-time coordinates of same events in different inertial reference frames are connected by affine transformation. It follows that a body that is inertial in the absolute inertial reference frame is inertial in any other inertial reference frame.

For any bodies $A$ and $B$ let us denote by $L_{BA}:O_B\rightarrow O_A$ the affine mapping that connects $O_B$ and $O_A$ such that each event $(x,\tau_B)$ in $O_B$ coincides with the event $L_{BA}(x,\tau_B)$ in $O_A$. Without loss of generality we assume that the origins of both reference frames $O_A$ and $O_B$ are the same: $x_{BA}(0)=0$ and $\tau_{BA}(0)=0$. Then the mapping $L_{BA}$ is linear.

\begin{lemma}\label{lemma_LBA}
The actual directions $\{e_i|1\le i\le m\}$ are eigenvectors of the mapping $L_{BA}$.
\end{lemma}
\begin{proof}
The directions of reference frame axes are imaginary directions in the event space. But the directions of the vectors $\{e_i|1\le i\le m\}$ in the absolute reference frame correspond to the only possible ``real'' motion directions of elementary bodies going from a graph vertex that coincide with the reference frame origin and therefore they do not depend on reference frames. It follows that the actual directions are invariant by any affine transformation of reference frames.
\end{proof}

In the following we demonstrate that the number $m$ of actual space-time directions is equal to the dimension $n+1$ of the event space $O$. Till then we choose $n+1$ linearly independent actual space-time directions and further we consider only the arcs in the graph that correspond to these directions. A fragment of an suitable $2$-dimensional environment is shown in Fig.~\ref{fig:2}.

\begin{figure}
\centering
\includegraphics[scale=0.3]{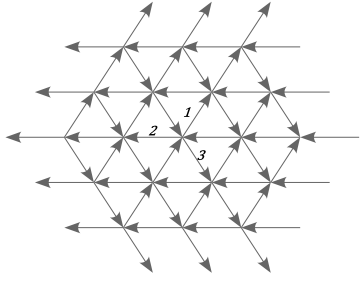}
\caption{2-dimensional environment with the actual spacial directions $D=\{\vec{1},\vec{2},\vec{3}\}$.}
\label{fig:2}
\end{figure}

Therefore, we can use the coordinates $(\lambda_1,\lambda_2,\ldots,\lambda_{n+1})$ in the reference frame $Q$ along with the coordinates $(x_1,\ldots,x_n,t)$ in the reference frame $O$, and also introduce the notion of actual reference frame $Q_A$ associated with the body $A$.

Let $M:Q\rightarrow O$ be the mapping that connects $Q$ and $O$ such that an event $(\lambda_1,\ldots,\lambda_{n+1})$ in $O$ coincides with the event $M(\lambda_1,\ldots,\lambda_{n+1})$ in $O$. Similarly, let $\Lambda_{BA}:Q_B\rightarrow Q_A$ be the mapping that connects $Q_B$ and $O_A$. From the definitions follows

\begin{theorem}\label{theorem:1}
$L_{BA}=M\cdot\Lambda_{BA}\cdot M^{-1}$.
\end{theorem}

Because the actual directions are basis vectors in $Q$, the matrix of the transformation $\Lambda_{AB}$ is diagonal.

Now, to uniquely determine the state transition measure $\tau_{BA}$ it is sufficient to determine $M$ and $\Lambda_{BA}$.
By definition, for example, we can assume that the coordinate $(x_1,\ldots,x_n,\tau_A)=(0,0, ..., 0,1)$ in the frame $O_A$ coincides with the coordinate $(\lambda_1,\lambda_2,\ldots,\lambda_{n+1}) = (1,1, ..., 1)$ in the reference frame $Q_A$. It means that we assume $\sum_{\vec{i}\in D}\vec{i}=\vec{0}$ and redefine $e_i=(\vec{i},\frac{1}{n+1})$ in the reference frame $O$, $\vec{i}\in D$. We proceed similarly with the absolute reference frames $O$ and $Q$. Thus, the transformations $M$, $\Lambda_{BA}$ and $L_{BA}$ are uniquely defined and we have completely defined the measure of the external state transition in the absolute reference frame and in the inertial reference frames as well.

Because the transformation matrix $\Lambda_{BA}$ in the general case looks like
\begin{equation}
\left(
  \begin{array}{cccc}
    \lambda_1	& 0 				& \cdots & 0  \\
    0 				& \lambda_2	& \cdots & 0 \\
    \vdots 		& \vdots 		& \ddots & \vdots \\
    0 				& 0 				& \cdots & \lambda_{n+1} \\
  \end{array}
\right) \; ,
\end{equation}
where $\lambda_i$ can be pairwise different, then the number of eigenvectors $L_{BA}$ is exactly equal to $n+1$, and therefore the following theorem holds.

\begin{theorem}\label{theorem:2}
For the existence of affine transformations connecting inertial reference frames it is necessary that the number of actual directions is equal to $n+1$.
\end{theorem}

\begin{corollary}\label{corollary:1}
For the existence of affine transformations connecting inertial reference frames it is necessary that the outdegree of each vertex in $G$ is equal to $n+1$.
\end{corollary}

Now we give a definition of the \emph{internal state} of a body. Let there be a bijection $\phi:A\rightarrow B$ for bodies $A$ and $B$ such that for all $b\in A$ elementary bodies $b$ and $\phi(b)$ are isomorphic (of the same ``colour''). We say that $A$ at the moment of proper time $\tau_A$ and $B$ at the moment of proper time $\tau_B$ are \emph{affine isomorphic} iff $\{(\phi(b),x_{bA}(\tau_A)|b\in A\}$=$\{(b,x_{bB}(\tau_B)|b\in B\}$. 
\begin{definition}
Two inertial bodies are in the same internal state at some moments of their proper time iff they are affine isomorphic at their respective proper time.
\end{definition}

Internal state of an inertial body does not depend on its spacial velocity in the absolute reference frame. Thus, the external state of an inertial body can be seen as a combination of two components: the spacial velocity of the body and its internal state. Because the spacial velocity of inertial bodies is constant, then by definition we can state that the measure of external state transition and the measure of internal state transition are the same.

If we now consider the body as an automata-like computational structure, whose states are defined as the internal states, the seemingly natural question whether a body can determine its own absolute velocity is by definition an algorithmically unsolvable problem or a meaningless question. If body states are by definiton the external states, then the same question has no substance, since the external state always contains information about the absolute velocity.

\subsubsection*{Final remarks.}
In this paper we have generalised the approach developed in~\cite{kurg2010} to the $n$-dimensional case, $n\ge 1$, through the introduction of actual reference frames. In~\cite{kurg2010} the case of $1$-dimensional environment with examples of affine isomorphic bodies is considered in detail, and some corollaries of the approach about the relationship between computational $w_{AB}$ and dynamic $v_{AB}$ properties of the bodies in the form of the time dilation formula, the velocity-addition formula, the length contraction/extension formula, etc. are shown.



\subsubsection*{Acknowledgments.} The author acknowledges the useful discussions on this work with Dr. Valeriy Anatoljevich Kozlovskyy, Dr. Igor Sergeevich Grunsky and Dr. Igor Potapov.


\begin{thebibliography}{4}
%
\bibitem{pedestrian} Toffoli T., A pedestrian's introduction to spacetime crystallography. IBM J. Res. Dev. 48, 1 (Jan. 2004), 13--29.
%
\bibitem{poincare} H.~Poincar\'e, Science et m\'ethode (1908).
%
\bibitem{3} I.~S.~Grunskyy, A.~N.~Kurgansky, Dynamics of collective of automata in discrete environment // Tr. Inst. Prikl. Mat. Mekh, 15, 50--56 (2007) (in Russian).
%
\bibitem{4} O.~Kurganskyy, Dynamics of a ``body'' in information environment, The 10th International Conference ``Stability, Control and Rigid Bodies Dynamics'' (ICSCD'08). - Donetsk, Ukraine, IAMM NASU, 2008, p.59. 
%
\bibitem{kurg2010} O.~Kurgansky, A state of a dynamic computational structure distributed in an environment: a model and its corollaries // eprint arXiv:1007.3836, 1-11, 2010
%
\end{thebibliography}
\end{document}